%% file: main.tex
\documentclass[runningheads]{llncs}
\usepackage[T1]{fontenc}
\usepackage{graphicx}
\usepackage{multicol}
\usepackage{multirow}
\usepackage{array}
\usepackage{amsmath,amssymb,amsfonts}
\usepackage{xcolor}
\usepackage[utf8]{inputenc}
\usepackage{amsfonts}
\usepackage{url}
\usepackage{babel}
\usepackage{mathrsfs}
\usepackage{caption}
\usepackage{subcaption}
\usepackage{listings}
\usepackage{hyperref}
\usepackage{thmtools}
\usepackage{thm-restate}
\usepackage{comment}
\usepackage{paralist}
\input{macro}
\begin{document}
\title{\polyqent: A Polynomial Quantified Entailment Solver}
%
%
\author{Krishnendu Chatterjee\inst{1} \and
Amir Kafshdar Goharshady\inst{2} \and
Ehsan Kafshdar Goharshady\inst{1} \and Mehrdad Karrabi \inst{1} \and Milad Saadat \inst{3} \and Maximilian Seeliger \inst{4} \and {\DJ}or{\dj}e \v{Z}ikeli\'{c} \inst{5}}
\authorrunning{K.~Chatterjee et al.}
%
\institute{Institute of Science and Technology Austria, Klosterneuburg, Austria 
\and
University of Oxford, Oxford, UK\\
\and 
Sharif University of Technology, Tehran, Iran\\
\and
Vienna University of Technology\\
\and
Singapore Management University, Singapore, Singapore\\ 
}
\maketitle              
\begin{abstract}
Polynomial quantified entailments with existentially and universally quantified 
variables arise in many problems of verification and program analysis. 
We present \polyqent\ which is a tool for solving polynomial quantified entailments in which variables on both sides of the implication are real valued or unbounded integers.
Our tool provides a unified framework for polynomial quantified entailment problems that arise in several papers in the literature. 
Our experimental evaluation over a wide range of benchmarks shows the 
applicability of the tool as well as its benefits as opposed to simply using existing SMT solvers to solve such constraints.

\end{abstract}

\input{Introduction}
\input{toolOverviewNew}

\input{evaluation}
\input{conclusion}

%
%
%
\bibliographystyle{splncs04}
\bibliography{bibliography}

\newpage
\appendix
\input{appendix}

\end{document}

%% file: macro.tex
\newcommand{\R}[0]{\mathbb{R}}

\newcolumntype{P}[1]{>{\centering\arraybackslash}p{#1}}
\newcommand{\satheu}{H1}
\newcommand{\coreheu}{H2}
\newcommand{\polyqent}{\textsc{PolyQEnt}}

\definecolor{codegreen}{rgb}{0,0.6,0}
\definecolor{codegray}{rgb}{0.5,0.5,0.5}
\definecolor{codepurple}{rgb}{0.58,0,0.82}

\lstdefinestyle{mystyle}
{
	language = python,
	basicstyle = {\ttfamily\footnotesize},
	stringstyle = {\color{string-color}},
	keywordstyle = {\bfseries\color{codegreen}},
	keywordstyle = [2]{\bfseries\color{codegreen}},
	keywordstyle = [3]{\bfseries\color{codepurple}},
	keywordstyle = [4]{\bfseries\color{teal}},
	otherkeywords = {;,<<,>>,++},
	morekeywords = [2]{do, then, fi, done},
	morekeywords = [3]{true, false, skip},
	morekeywords = [4]{Invariant},
}

\newcommand{\blueit}[1]{\textcolor{blue}{#1}}

\newcommand{\monoid}{\text{\(\mathcal{M}onoid\)}}

\renewcommand{\paragraph}[1]{\smallskip\noindent{\bf #1}}

%% file: introduction.tex
\section{Introduction}\label{sec:intro}

\paragraph{Polynomial constraint solving in verification.} A fundamental computational task that arises in several contexts of 
verification and static analysis of programs  is constraint solving over 
polynomials. The most prominent example of an application in program analysis is {\em template-based synthesis}~\cite{GulwaniSV08}. Given a program and a property, a classical approach to proving that the program satisfies the property is to compute a certificate (i.e.~a formal proof) of the property~\cite{floyd1993assigning}. This can be achieved by fixing a suitable {\em symbolic template} for the certificate, which allows reducing the program verification problem to computing values of symbolic template variables that together give rise to a correct certificate~\cite{SrivastavaGF13}. Such an approach with symbolic templates being linear or polynomial functions has found extensive applications in static analysis of programs with linear or polynomial arithmetic, including termination analysis~\cite{ColonS01,PodelskiR04,Chatterjee16}, invariant generation~\cite{Colon03,FengZJZX17,Chatterjee20}, reachability~\cite{AsadiC0GM21}, cost analysis~\cite{0002AH12,GulwaniZ10,Zikelic22}, program synthesis~\cite{GulwaniJTV11,Goharshady23} and probabilistic program analysis~\cite{ChakarovS13,Chatterjee16,ChatterjeeNZ17}. 
This approach has also found extensive applications in other domains of computer science, e.g.~controller verification and synthesis~\cite{AhmadiM16,prajna2002introducing,PrajnaJ04}.

\paragraph{Polynomial Quantified Entailments.} In all cases mentioned above, the goal of template-based synthesis is to compute a certificate for the property of interest, where the certificate is computed in the form of a symbolic linear or polynomial function. The computation is achieved by a reduction to solving a system of {\em polynomial entailments}, i.e.~a system of $K \in \mathbb{N}$ constraints of the form
\begin{equation*}
    \exists t \in \mathbb{R}^m.\,\, \bigwedge_{i=1}^K \Big( \forall x \in \mathbb{R}^n.\,\, \Phi^i(x,t) \Longrightarrow \Psi^i(x,t)\Big).
\end{equation*}
Here, the variables $t \in \mathbb{R}^m$ present real-valued {\em template coefficients} of the symbolic linear or polynomial function that together define the certificate, and each $\Phi^i$ and $\Psi^i$ is a {\em boolean combination of polynomial inequalities} over a vector $x \in \mathbb{R}^n$ of program variables. The entailments $\forall x \in \mathbb{R}^n.\,\, \Phi^i(x,t) \Longrightarrow \Psi^i(x,t)$ together encode the necessary properties for the symbolic template polynomial to define a correct certificate. Hence, any valuation of the variables $t \in \mathbb{R}^m$ that gives rise to a solution to the system of constraints above also gives rise to a concrete instance of the correct certificate. We refer to each entailment $\forall x \in \mathbb{R}^n.\,\, \Phi^i(x,t) \Longrightarrow \Psi^i(x,t)$ as a {\em polynomial quantified entailment (PQE)}, and to the problem of solving a system of PQEs in eq.~\eqref{eq:pqe} as {\em PQE solving}.


\paragraph{Solving PQEs via positivity theorems.} Initial work on template-based synthesis has focused on linear programs and linear certificate templates. A classical approach to solving this problem is to use Farkas' lemma that considers implications over linear expressions~\cite{farkas02}, which has been applied in several works related to program analysis, e.g.~\cite{Colon03,Heizmann14,ChakarovS13,Chatterjee18}. However, this method was insufficient for analyzing programs described by polynomials, e.g.~programs that may contain program variable multiplication. A generalization of Farkas' lemma-style reasoning to the setting of polynomial constraints and PQEs is achieved by using {\em positivity theorems}, such as Handelman's~\cite{handelman88} and Putinar's theorem~\cite{putinar93}. It was shown in~\cite{Chatterjee16,AsadiC0GM21} that they can be applied towards effectively solving systems of PQEs, with applications in static analysis of polynomial programs for termination~\cite{Chatterjee16}, reachability~\cite{AsadiC0GM21}, invariant generation~\cite{Chatterjee20}, non-termination~\cite{Chatterjee21} properties and for probabilistic program analysis~\cite{ChatterjeeNZ17,Wang19}.

\paragraph{\polyqent.} PQE solving via positivity theorems is becoming increasingly popular in static program analysis, however tool support for its integration into these analyses is non-existent and researchers have relied on their own implementations. In this work, we present our tool \polyqent\ which implements methods for solving systems of PQEs over the theories of polynomial real or unbounded integer arithmetic, based on Handelman's theorem, Putinar's theorem and Farkas' lemma. We provide efficient implementations of each of these methods together with practical heuristics that we observed to improve their performance. At the same time, our tool preserves soundness and relative completeness guarantees of these translations as established in the previous results in the literature~\cite{AsadiC0GM21,Chatterjee20}. 
We envision that \polyqent\ will allow future research to focus on the design of appropriate certificate templates, whereas the constraint solving part can be fully delegated to our tool. \polyqent\ is implemented in Python and publicly available 
on GitHub \footnote{\href{https://github.com/ChatterjeeGroup-ISTA/polyqent}{https://github.com/ChatterjeeGroup-ISTA/polyqent}}.
It allows users to provide constraints as input in the SMT-LIB syntax~\cite{BarFT-SMTLIB}, which is a standard and widely used input format. \polyqent\ also automates the selection of the positivity theorem to be used (Handelman's theorem, Putinar's theorem or Farkas' lemma) in order to achieve most efficient constraint solving while providing the soundness and relative completeness guarantees.


\paragraph{Experimental evaluation.} 
We experimentally evaluate \polyqent\ on several benchmarks collected from the literature on termination and
non-termination analysis in polynomial programs, termination of probabilistic programs and polynomial program synthesis. While all these problems could also be directly solved using off-the-shelf SMT solvers that support quantifier elimination, e.g.~Z3~\cite{de2008z3}, our experimental results show {\em significant improvements in runtime} when positivity theorems are used to eliminate quantifier alternation.

\paragraph{Comparison to constrained Horn clauses.} The problem of PQE solving syntactically resembles the more studied problem of constrained Horn clause (CHC) solving. CHC solving is a classical approach to program verification~\cite{Grebenshchikov12} with many readily available tools, e.g.~\cite{GurfinkelKKN15,HojjatR18,KomuravelliGC16,abs-1907-03998,FedyukovichPMG18}. However, the goal of the PQE solving problem is fundamentally different from CHC solving, and methods for one problem are not readily applicable to the other problem. In CHC solving, the focus is on computing boolean predicates that together make the CHC valid. In contrast, template-based synthesis applications discussed above require computing {\em values of template variables} that together define a certificate conforming to a given template, where the template is specified as a boolean combination of symbolic linear or polynomial inequalities over program variables. Hence, what would be viewed as an uninterpreted predicate in CHC solving, becomes a fixed boolean combination of polynomial inequalities of a specified maximal polynomial degree in PQE solving. The existing CHC solvers are thus not applicable to the problem of PQE solving.

Finally, FreqHorn~\cite{FedyukovichPMG19} is a CHC solver that is able to generate interpretations for uninterpreted predicates in the form of polynomial inequalities. Once predicate interpretations are generated, the satisfiability of the resulting interpreted formula is checked by performing quantifier elimination via a technique based on Model-Based Projections~\cite{FedyukovichPMG18}, after which the resulting quantifier-free formula is solved via Gauss-Jordan elimination. Hence, the last step of FreqHorn’s methodology is also applicable to our problem. The novelty provided by our PolyQEnt is twofold. First, our quantifier elimination procedure is based on positivity theorems, which have recently been extensively utilized in program analysis and verification but for which tooling support is inexistent. Second, while the methodology of FreqHorn is in principle applicable to our problem setting, the tool itself only supports CHC solving over uninterpreted predicates. Hence, we could not perform a direct comparison.

%% file: toolOverviewNew.tex
\section{Tool Overview}\label{sec:overview}

In this section we provide an overview of our tool and discuss details of the considered problem, the tool's architecture and its backend design.


\subsection{Problem Statement}\label{sec:problem}

\def\bluecolor{\color{blue}}
\def\blackcolor{\color{black}}

The problem of {\em polynomial quantified entailment (PQE) solving} is concerned with computing a valuation of existentially quantified variables $t_1,\dots,t_m$ that make the following logical formula true
\begin{equation}\label{eq:pqe}
    \exists t \in \mathbb{R}^m.\,\, \bigwedge_{i=1}^K \Big( \forall x \in \mathbb{R}^n.\,\, \Phi^i(x,t) \Longrightarrow \Psi^i(x,t)\Big).
\end{equation}
Here, each $\Phi^i$ and $\Psi^i$ is a boolean combination of polynomial inequalities of the form $p(t_1,\dots,t_m,x_1,\dots,x_n) \bowtie 0$, with $p$ a polynomial function and $\bowtie\, \in \{\geq,>\}$. We refer to each entailment $\forall x \in \mathbb{R}^n.\,\, \Phi^i(x,t) \Longrightarrow \Psi^i(x,t)$ as a {\em polynomial quantified entailment (PQE)}, and to the formula in eq.~\eqref{eq:pqe} as a {\em system of PQEs}. 

In what follows, we consider systems of PQEs defined over the background theory of {\em real arithmetic}. However, our \polyqent\ is also applicable to PQEs defined over {\em unbounded integer arithmetic} (i.e.~mathematical integers). While our presentation will mostly focus on PQEs defined over real arithmetic, in Section~\ref{sec:backend} we discuss differences that arise in considering unbounded integer arithmetic and how \polyqent\ addresses them.

\paragraph{Canonical form of PQEs.} We say that a PQE $\forall x \in \mathbb{R}^n.\,\, \Phi(x,t) \Longrightarrow \Psi(x,t)$ is in the {\em canonical form}, if $\Phi$ is a conjunction of finitely many polynomial inequalities and $\Psi$ is a single polynomial inequality, i.e.~if $\Phi \equiv \bigwedge_{j=1}^{n} \Big(p_j(t_1,\dots,t_m,x_1,\dots,x_n) \bowtie 0 \Big)$ and $\Psi \equiv p(t_1,\dots,t_m,x_1,\dots,x_n) \bowtie 0.$
Each PQE can be translated into an equisatisfiable PQE in the canonical form, defined over the same set of free variables $t \in \mathbb{R}^m$ and universally quantified variables $x \in \mathbb{R}^n$. This is an important result, as this translation presents the preprocessing step of our \polyqent. The following proposition formally proves this claim. The proof, together with the procedure employed by \polyqent\ to achieve this translation, is provided in the Appendix.

\begin{restatable}{proposition}{canonical}
\label{prop:canonical}
    Each PQE can be translated into an equisatisfiable PQE in the canonical form, defined over the same sets of quantified variables.
\end{restatable}

\subsection{Tool Architecture}\label{sec:architecture}


\paragraph{Architecture.} 
The tool takes as input a system of PQEs in the form as in eq.~\eqref{eq:pqe}. The input is provided in the SMT-LIB format~\cite{BarFT-SMTLIB}, alongside with an optional config file in the \texttt{.json} format (see the following paragraph for details). Examples are provided in the tool's repository. Note that we do not assume any logical structure of the polynomial inequalities in the PQEs, i.e.~polynomial predicates in each PQE can have arbitrary \texttt{and}/\texttt{or} logical connectives.

The input files are then parsed and the PQEs are translated to their equisatisfiable canonical forms as in Proposition~\ref{prop:canonical}. \polyqent\ then applies the appropriate positivity theorem to reduce the problem of PQE solving to solving a fully existentially quantified system of polynomial constraints (see Section~\ref{sec:backend} for details). The resulting fully existentially quantified system of polynomial constraints is then fed to an SMT-solver. In case when the ``UNSAT Core" heuristic is used, \polyqent\ will further process the output of the SMT-solver (see Section~\ref{sec:backend} for details). Finally, the output of \polyqent\ is either (1) \texttt{SAT} with a valuation of existentially quantified variables for which the system of PQEs is valid, (2) \texttt{UNSAT} if the SMT-solver proves unsatisfiability, or (3) \texttt{Unknown} if the SMT-solver returns unknown.


\paragraph{Configuration file (optional).}
  \polyqent\ has a default (and recommended) configuration, which does not require the user to provide the config file. However, we also allow the user to change some of the parameter values used by \polyqent\ (e.g. the positivity theorems used, the technical parameters, the SMT-solver, or the heuristics) by providing a \texttt{.json} config file. The technical details can be found in the tool's readme file and Appendix \ref{sec:config}.


\begin{remark}
    We integrated the commands for running Z3 and MathSAT5 into \polyqent. However, \polyqent\ also stores the obtained system of existentially quantified polynomial inequalities in an SMTLIB format output file which can then be fed to other SMT-solvers.
\end{remark}



\subsection{Backend Algorithms and Heuristics}\label{sec:backend}

We now overview the backend of our tool. Observe that the system of PQEs in eq.~\eqref{eq:pqe} contains quantifier alternation with existential quantification preceding universal quantification. As mentioned in Section~\ref{sec:intro}, constraints involving such quantifier alternation can in principle be solved directly by using an off-the-shelf SMT solver that supports quantifier elimination, e.g.~Z3~\cite{de2008z3}. However, decision procedures for solving such constraints are known to be highly unscalable and a major source of inefficiency~\cite{RENEGAR1992329}. PQE solving can be made {\em significantly more efficient} by first using positivity theorems to eliminate universal quantification and reduce the problem to solving a purely existentially quantified system of polynomial constraints. Our experiments in Section~\ref{sec:evaluation} support this claim. In what follows, we describe an overview of positivity theorems and also how \polyqent\ uses them for quantifier elimination. 

\paragraph{Overview of positivity theorems.} All three positivity theorems implemented in \polyqent, namely Handelman's theorem, Putinar's theorem and Farkas' lemma, consider constraints of the form
\begin{equation}\label{eq:implication}
\begin{split}
    \forall x \in \mathbb{R}^n. &\Big(f_1(x) \bowtie 0 \land \dots \land f_m(x) \bowtie 0 \Longrightarrow g(x) \bowtie 0 \Big),
\end{split}
\end{equation}
where $f_1, \dots, f_m,g$ are polynomials over real-valued variables $x_1,\dots,x_n$ and each $\bowtie \,\in \{\geq,>\}$ (we discuss the unbounded integer variables case in the following paragraph). These theorems provide {\em sound} translations of the implication in eq.~\eqref{eq:implication} into a purely existentially quantified system of polynomial inequalities over newly introduced auxiliary symbolic variables. Translations are sound in the sense that, if the obtained purely existentially quantified system of constraints is satisfiable, then the implication in eq.~\eqref{eq:implication} is valid. Moreover, these theorems provide additional sufficient conditions under which the translation is also {\em complete}, i.e.~the implication in eq.~\eqref{eq:implication} and the resulting existentially quantified system are equisatisfiable:
\begin{itemize}
    \item \textbf{Farkas' lemma.} Farkas' lemma~\cite{farkas02,matouvsek07} provides a sound and complete translation if all $f_i$'s and $g$ are linear expressions.
    \item \textbf{Handelman's theorem.} Handelman's theorem~\cite{handelman88} provides a sound translation if all $f_i$'s are linear expressions and $g$ is a polynomial expression. Moreover, the translation can be made complete if in addition all inequalities $f_i(x_1,\dots,x_n) \bowtie 0$ are non-strict, the inequality $g(x_1,\dots,x_n) \bowtie 0$ is strict and the set $\{(x_1,\dots,x_n)\mid \forall 1\leq i\leq n.\, f_i(x_1,\dots,x_n) \geq 0\} \subseteq \mathbb{R}^n$ is bounded.
    \item \textbf{Putinar's Theorem.} Putinar's theorem~\cite{putinar93} provides a sound translation if all $f_i$'s and $g$ are polynomial expressions. Moreover, the translation can be made complete if in addition all inequalities $f_i(x_1,\dots,x_n) \bowtie 0$ are non-strict, the inequality $g(x_1,\dots,x_n) \bowtie 0$ is strict and the set $\{(x_1,\dots,x_n)\mid \forall 1\leq i\leq n.\, f_i(x_1,\dots,x_n) \geq 0\} \subseteq \mathbb{R}^n$ is bounded. 
\end{itemize}

\paragraph{Polynomial unbounded integer arithmetic.} While positivity theorems consider polynomials over real-valued variables, the resulting translations remain sound under unbounded integer arithmetic. In this case, strict inequalities can always be treated as non-strict by incrementing the appropriate side of the inequality by $1$, and the above yield sound but incomplete translations. The usage of positivity theorems in polynomial unbounded integer arithmetic was discussed in~\cite{Chatterjee21}. \polyqent\ provides support for unbounded integer arithmetic PQEs.

\paragraph{Positivity theorems in \polyqent.} \polyqent\ implements the translations via Farkas' lemma, Handelman's theorem and Putinar's theorem, and uses them to eliminate quantifier alternation in the system of PQEs. When applied to a PQE in the canonical form, the translation yields a purely existentially quantified system of polynomial constraints. Hence, the system of PQEs is translated into a purely existentially quantified formula, with the first-order variables being the existentially quantified variables $t \in \mathbb{R}^m$ as in eq.~\eqref{eq:pqe} as well as new symbolic variables introduced by the translation. 
For each PQE, the default configuration of \polyqent\ automatically chooses the positivity theorem that leads to a most efficient translation, while satisfying the soundness and relative completeness requirements of each theorem listed above, with Farkas' Lemma being the most efficient, followed by Handelman's theorem, and finally Putinar's Theorem. Alternatively, by providing the \texttt{.json} config file as described in Section~\ref{sec:architecture}, the user can opt for a different positivity theorem to be used.

In Appendix~\ref{subsec:geometry}, we provide formal statements of each positivity theorem that \polyqent\ implements. When invoked in \polyqent\ each theorem has some input parameters which are set to default values, that can be modified in the config file (See Appendix \ref{subsec:geometry} for details). 

\paragraph{Heuristics.} We conclude by outlining two heuristics that we implemented in \polyqent\ and that we observed to improve the tool's performance. The effect of each heuristic is studied in our experimental evaluation in Section~\ref{sec:evaluation}:
\begin{enumerate}
    \item {\em Assume-SAT.} For a PQE to be valid, either (1)~the left-hand-side of the entailment needs to be satisfiable and to imply the right-hand-side at all satisfying points, or (2)~the left-hand-side of the entailment needs to be unsatisfiable. Positivity theorems translate a system of PQEs into a purely existentially quantified system of polynomial constraints, whose satisfiability implies either (1) or (2). The Assume-SAT heuristic instead collects a system of constraints whose satisfiability only implies (1). This heuristic is sound but it leads to incompleteness, as (2) also implies that the system of PQEs is satisfiable. However, we observed that the Assume-SAT heuristic can sometimes considerably reduce the size of the obtained system of constraints, which can make the subsequent SMT solving step significantly more efficient. 
    \item {\em UNSAT core.} This heuristic was proposed in~\cite{Goharshady23}, a work which uses positivity theorems for program synthesis. Since the positivity theorem translations introduce a large number of fresh symbolic variables that are now existentially quantified, the idea behind the heuristic is to first try to solve the resulting system of constraints while adding additional constraints that set the values of some of these newly introduced symbolic variables to $0$. If the SMT-solver returns \texttt{SAT}, then the original system is satisfiable as well. Otherwise, SMT solvers such as Z3~\cite{de2008z3} and MathSAT~\cite{cimatti2013mathsat5} can return an unsatifsiability core, a subset of constraints that are unsatisfiable themselves. If the core contains none of the newly added constraints, it implies that the original system was unsatifiable. Otherwise, \polyqent\ removes the newly added $t=0$ constraints that are in the core and repeats this procedure.
\end{enumerate}

%% file: evaluation.tex
\section{Experimental Evaluation}\label{sec:evaluation}

We evaluate the performance of our tool on three benchmark sets in the following three subsections. The goal of our experiments is to illustrate (1)~soundness of the tool, (2)~its ability to solve PQEs that arise in program analysis literature, (3)~the necessity of using positivity theorems for quantifier elimination as opposed to feeding PQEs to an SMT solver directly, and (4)~to study the performance of different combinations of our two heuristics and different SMT solvers. Benchmarks are provided in the $\texttt{.smt2}$ format. All the experiments were conducted on a Debian 11 machine with AMD EPYC 9654 
2.40 GHz CPU and 6 GB RAM with a timeout of 180 seconds.

\paragraph{Baselines.} In order to illustrate the necessity of using positivity theorems for quantifier elimination, on all three benchmark sets we compare \polyqent\ against baseline methods which directly use Z3(v4.13.4)~\cite{de2008z3}  and CVC5(v1.2.0)~\cite{cvc5} to solve the system of PQEs, i.e.~without separately performing quantifier elimination.

\subsection{Termination and Non-Termination}\label{sec:experterm}

\noindent The first benchmark set consists of systems of PQEs that arise in termination analysis of programs. We consider TermComp'23~\cite{termcomp}, C-Integer category, benchmark suite that consists of $335$ non-recursive programs written in C. Initial value of every program variable is assumed to be in the range $[-1024,1023]$. The goal is to either prove termination or non-termination of each program.

\paragraph{Extraction of PQEs.} We extract PQEs for termination and non-termination proving as follows:

For termination proving, we use the template-based synthesis method for computing ranking functions~\cite{AsadiC0GM21}. For each program, as is common in termination analysis, we first used ASPIC~\cite{FeautrierG10} to generate a supporting invariant with respect to which ranking function is to be computed. We then use the method of~\cite{AsadiC0GM21} to extract two systems of PQEs that each encode the ranking function synthesis problem -- one for the linear (polynomial degree 1) and one for the quadratic (polynomial degree 2) templates. Finally, we run \polyqent\ with Farkas' lemma to solve the first system of PQEs, and with Putinar's theorem to solve the second system of PQEs, respectively.

For non-termination proving, we use the template-based synthesis method for computing a non-termination certificate of~\cite{Chatterjee21}. We use the same parameters as in the `best config' setting of the artifact of~\cite{Chatterjee21} to define the template for the non-termination certificate. We run the method of~\cite{Chatterjee21} to generate the system of PQEs that encodes the non-termination certificate synthesis problem. Finally, to solve the resulting system of PQEs, we run \polyqent\ with Farkas' lemma in 101/103 cases where `best config' prescribes linear templates, and with Putinar's theorem in 2/103 cases where it prescribes polynomial templates.
\paragraph{Results.} Table~\ref{table:termcomp} shows a summary of our results on \polyqent's performance on PQEs coming from termination analysis. Clearly, \polyqent\ performs far better than using Z3 or CVC5 directly for quantifier elimination. Comparing performance of different heuristics, we observe several interesting points:
\begin{itemize}
    \item Applying heuristic \texttt{\satheu} (Assume-SAT) helps \polyqent\ greatly, especially when using MathSAT as the backend solver. 
    \item Applying heuristic \texttt{\coreheu} (UNSAT Core) does not prove any unique cases that other configurations of \polyqent\ could not prove. However, Z3 solves more instances than MathSAT when both are equipped with this heuristic. 
    \item While Z3 solves more instances of PQEs coming from termination analysis, MathSAT outperforms it in solving PQEs coming from non-termination analysis. This suggests running several SMT solvers in parallel in order to achieve the best results. 
\end{itemize}

\paragraph{Runtime comparison.} Fig.~\ref{fig:termcomp-plots} plots the number of instances solved by each tool against runtime. For several benchmarks the direct-Z3 method terminates nearly instantly. 
Apart from them, it can be seen that Z3 is faster than MathSAT in settings where \texttt{\satheu} heuristic is not applied, but applying \texttt{\satheu} makes MathSAT slightly more efficient than Z3. Moreover, compared to not using any heuristics, applying \texttt{H1} and \texttt{H2} results in a speed-up in 220 and 32 benchmarks, respectively.

\input{table-termcomp}

\begin{figure}[t]
    
    \begin{center}
    \hspace{-2mm}
    \begin{subfigure}[c]{0.50\textwidth}
    \includegraphics[width=\textwidth,trim={0.2cm 0 0.3cm 0},clip]{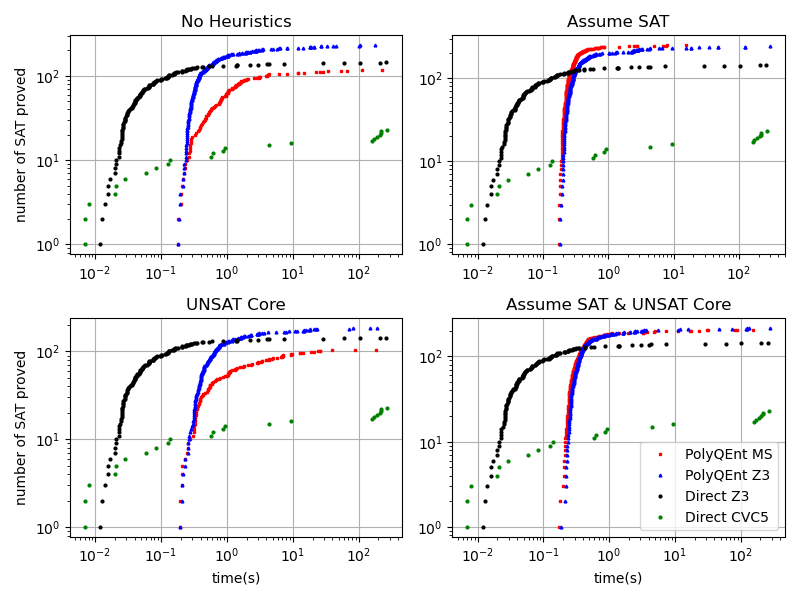}
    \end{subfigure}
    \hspace{5mm}
    \begin{subfigure}[c]{0.45\textwidth}
    \includegraphics[width=\textwidth, trim={0.6cm 0 1.6cm 0},clip]{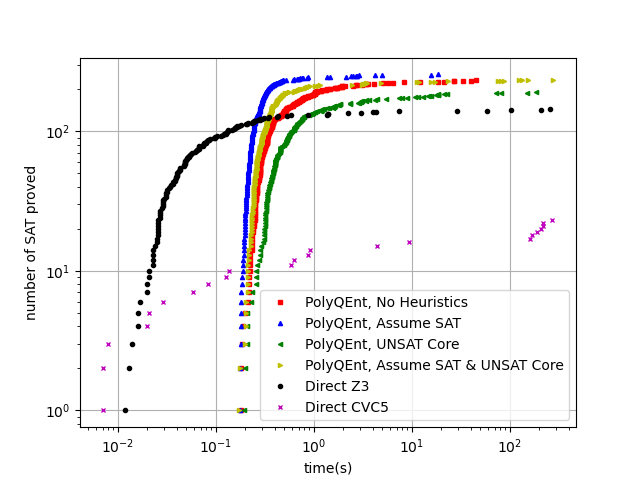}
    \end{subfigure}
    \end{center}
    \vspace{-6mm}
    \caption{Performance of \polyqent\ with different settings in comparison to baselines. Both axes are scaled logarithmically for better visualization. The leftmost four plots demonstrate the effect of using different solvers and heuristic settings, and the rightmost plot unionizes solvers to just compare heuristics.}
    \label{fig:termcomp-plots}
\end{figure}


\subsection{Almost-Sure Termination}

The second benchmark set comes from almost-sure termination proving in probabilistic programs. We collected benchmarks from two sources: i) 10 benchmarks from~\cite{Wang19} (Table 3) and ii) 7 benchmarks from~\cite{Kura19}. We choose these benchmarks because both works consider probabilistic program models of different applications which are required to be almost-surely terminating.

\paragraph{Extraction of PQEs.} To prove almost-sure termination, we use the template-based synthesis method for computing ranking supermartingales~\cite{Chatterjee16}. For each program, as is common in termination analysis, we used ASPIC~\cite{FeautrierG10} and StInG~\cite{SankaranarayananSM04} to generate supporting invariants with respect to which ranking supermartingale is to be computed. We then fix a quadratic (polynomial degree~2) template for the ranking supermartingale and use the method of~\cite{Chatterjee16} to extract a system of PQEs that encodes the synthesis problem. Finally, we use \polyqent\ with Handelman's theorem to solve the system of PQEs.

\paragraph{Results.} Table~\ref{table:AST} shows a summary of our results on \polyqent's performance on PQEs coming from almost-sure termination analysis. Applying Z3 or CVC5 directly is not successful on any of the benchmarks. On the other hand, \polyqent\ successfully solves 13 out of 17 instances.
Moreover, it can be seen that in most of the settings, running one SMT-solver alone does not provide the best results. This again suggests running several SMT solvers in parallel. 





\input{table-AST}
\subsection{Synthesis}

The third benchmark set comes from polynomial program synthesis, where we collect $32$ benchmarks from \textit{PolySynth}~\cite{Goharshady23}. Each benchmark is a non-deterministic program that contains holes and a desired specification. PolySynth uses a template-based technique to synthesize suitable polynomial expressions for the holes such that the specification is satisfied.  

\paragraph{Extraction of PQEs.} We first use the method of~\cite{Goharshady23} to extract a system of PQEs for polynomial program synthesis. For each benchmark, we ran \polyqent\ using Farkas lemma as well as Handelman's and Putinar's theorem with polynomial degree~2 templates. For brevity, we present a union of these runs, where we consider the faster setting whenever several of them worked. 

\paragraph{Results.} Table~\ref{table:synth} shows a summary of our results on \polyqent's performance on PQEs coming from program synthesis. We note that for two benchmarks (namely, \texttt{positive\_square\_with\_holes} and \texttt{positive\_square\_with\_number\_holes}) the Direct Z3 method could find a solution while \polyqent\ equipped with Farkas and Handelman could not. However, using Putinar with polynomial degree 2, \polyqent\ can solve those instances as well. Other than that, \polyqent\ outperforms the Direct Z3 and Direct CVC5 approaches both in terms of the number of solved instances and runtime. Comparing the performance of heuristics, it can be seen that the UNSAT Core heuristic slightly outperforms the Assume-SAT heuristic. 

\input{table-synth}

%% file: table-termcomp.tex
\begin{table}[t]
	\centering
	\texttt{
        \resizebox{\textwidth}{!}
		{
			\fontsize{6pt}{7pt}\selectfont
                \begin{tabular}{|P{18mm}|P{5mm}|P{5mm}|P{5mm}|P{5mm}|P{5mm}|P{5mm}|P{5mm}|P{5mm}|P{5mm}|P{5mm}|P{5mm}|P{5mm}|P{10mm}|P{10mm}|}
                    \hline 
                     \multirow{2}{*}{Specification} & \multicolumn{3}{c|}{Base} & \multicolumn{3}{c|}{Base+\satheu} & \multicolumn{3}{c|}{Base+\coreheu} & \multicolumn{3}{c|}{Base+\satheu+\coreheu} & \multirow{2}{*}{Dir. Z3} & \multirow{2}{*}{Dir. CVC5}\\
                     \cline{2-13}
                     & MS 5 & Z3 & U. & MS 5 & Z3 & U. & MS 5 & Z3 & U.  & MS 5 & Z3& U. & &
                     \\
                     \hline 
                     Termination & 44 & 154 & 154 & 148 & 153 & 153 & 43 & 109 & 113 & 125 & 132 & 145 & 103 & 9 \\
                     \hline 
                     Non-Termination & 72 & 76 & 78 & 101 & 87 & 103 & 61 & 75 & 76 & 81 & 82 & 89 & 39 & 9 \\
                     \hline 
                     \hline 
                     Avg. Time (s) & 4.4 & 3.9   & 1.6 &	0.5 &	2.5 &	0.5 &	4.3 &	3.5 &	3.4 &	2.7 &	3.7 &	3.6 &	1.6 & 19.4 \\
                     \hline
                \end{tabular}
            }
        }
    \caption{Summary of our results on TermComp benchmarks. The \texttt{Base} column shows the results of \polyqent\ without any heuristics. The next 3~columns enable heuristics \satheu (Assume-SAT) and \coreheu (UNSAT Core). The `Dir. Z3' and `Dir. CVC5' columns summarize the results obtained by applying Z3 and CVC5 directly. For each setting, we show the number of instances solved by MathSAT 5 (\texttt{MS 5}), Z3 and their union (\texttt{U.}).}
    \label{table:termcomp}
    \vspace{-4mm}
\end{table}

%% file: table-AST.tex
\begin{table}[t]
	\centering
	\texttt{
        \resizebox{\textwidth}{!}
		{
			\fontsize{6pt}{7pt}\selectfont
                \begin{tabular}{|P{18mm}|P{5mm}|P{5mm}|P{5mm}|P{5mm}|P{5mm}|P{5mm}|P{5mm}|P{5mm}|P{5mm}|P{5mm}|P{5mm}|P{5mm}|P{10mm}|P{10mm}|}
                    \hline 
                     \multirow{2}{*}{Benchmark set} & \multicolumn{3}{c|}{Base} & \multicolumn{3}{c|}{Base+\satheu} & \multicolumn{3}{c|}{Base+\coreheu} & \multicolumn{3}{c|}{Base+\satheu+\coreheu} & \multirow{2}{*}{Dir. Z3} & \multirow{2}{*}{Dir. CVC5}\\
                     \cline{2-13}
                     & MS 5 & Z3 & U. & MS 5 & Z3 & U. & MS 5 & Z3 & U.  & MS 5 & Z3 & U. & &
                     \\
                     \hline 
                     From \cite{Wang19} & 7 & 6 & 7 & 7 & 6 & 7 & 7 & 6 & 7 & 7 & 5 & 7 & 0 & 0 \\
                     \hline 
                     From \cite{Kura19} & 6 & 5 & 6 & 5 & 6 & 6 & 5 & 4 & 5 & 6 & 4 & 6 & 0 & 0\\
                     \hline 
                     \hline 
                     Avg. Time (s) & 10.4 &	6.1 &	10.4 &	9.5 &	3.5 &	9.7 &	10.1 &	3.4 &	10.1 &	10.9 &	2.0 &	10.9 & NA & NA \\
                     \hline 
                \end{tabular}
            }
        }
        \caption{Results on the second benchmark set.}
        \label{table:AST}
        \vspace{-1em}
\end{table}


%% file: table-synth.tex
\begin{table}[t]
	\centering
\vspace{-4mm}
	\texttt{
        \resizebox{\textwidth}{!}
		{
			\fontsize{6pt}{7pt}\selectfont
                \begin{tabular}{|P{18mm}|P{5mm}|P{5mm}|P{5mm}|P{5mm}|P{5mm}|P{5mm}|P{5mm}|P{5mm}|P{5mm}|P{5mm}|P{5mm}|P{5mm}|P{10mm}|P{10mm}|}
                    \hline 
                     \multirow{2}{*}{Benchmark Set} & \multicolumn{3}{c|}{Base} & \multicolumn{3}{c|}{Base+\satheu} & \multicolumn{3}{c|}{Base+\coreheu} & \multicolumn{3}{c|}{Base+\satheu+\coreheu} & \multirow{2}{*}{Dir. Z3} & \multirow{2}{*}{Dir. CVC5}\\
                     \cline{2-13}
                     \cline{2-13}
                     & MS 5 & Z3 & U. & MS 5 & Z3 & U. & MS 5 & Z3 & U. & MS 5 & Z3 & U. & &
                     \\
                     \hline 
                     From \cite{Goharshady23} &  28 & 30 & 30 & 28 & 29 & 29 & 28 & 30 & 30 & 28 & 29 & 29 & 24 & 0 \\
                     \hline 
                     \hline 
                     Avg. Time (s) & 3.2 & 2.6 & 2.6 & 2.4 & 2.4 & 2.4 & 2.9 & 2.7 & 2.6 & 2.6 & 2.5 & 2.5 & 14.0 & NA \\
                     \hline 
                \end{tabular}
            }
        }
        \caption{Results on the third benchmark set.}
        \label{table:synth}
        \vspace{-9mm}
\end{table}


%% file: conclusion.tex
\paragraph{Concluding Remarks.}
We presented our tool \polyqent\ for solving systems of polynomial quantified entailments, a problem that arises in many template-based synthesis methods for program analysis and verification. The significance of \polyqent\ is that, for template-based synthesis, it separates the task of certificate design, which future research can focus on, and the task of polynomial constraint solving, for which \polyqent\ provides an efficient tool support. Future work includes studying further heuristics towards making \polyqent\ even more efficient for solving systems of polynomial quantified entailments.

%% file: appendix.tex
\section{Proof of Proposition \ref{prop:canonical}}

\canonical*

\begin{proof}
    Consider a system of PQEs as in eq.~\eqref{eq:pqe}
    \[ \exists t \in \mathbb{R}^m.\,\, \bigwedge_{i=1}^K \Big( \forall x \in \mathbb{R}^n.\,\, \Phi^i(x,t) \Longrightarrow \Psi^i(x,t)\Big), \]
    with each $\Phi^i$ and $\Psi^i$ being an arbitrary boolean combination of polynomial inequalities over $x$ and $t$. It is a classical result that every formula in propositional logic can be translated into an equivalent formula in conjunctive normal form (CNF). Hence, we can translate each $\Phi^i(x,t) \Longrightarrow \Psi^i(x,t)$ into an equisatisfiable formula of the form
    \[ \Theta^i(x,t) = \bigwedge_{j=1}^{m_i}\bigvee_{l=1}^{n_i} (p^i_{j,l}(x,t) \bowtie^i_{j,l}\, 0), \]
    where each $p^i_{j,l}(x,t) \bowtie^i_{j,l}\, 0$ is a polynomial inequality. We then have
    \begin{equation*}
    \begin{split}
        &\exists t \in \mathbb{R}^m.\,\, \bigwedge_{i=1}^K \Big( \forall x \in \mathbb{R}^n.\,\, \Phi^i(x,t) \Longrightarrow \Psi^i(x,t)\Big) \\
        \equiv & \exists t \in \mathbb{R}^m.\,\, \forall x \in \mathbb{R}^n.\,\, \bigwedge_{i=1}^K \Big( \Phi^i(x,t) \Longrightarrow \Psi^i(x,t)\Big) \\
        \equiv & \exists t \in \mathbb{R}^m.\,\, \forall x \in \mathbb{R}^n.\,\, \bigwedge_{i=1}^K \Big(\bigwedge_{j=1}^{m_i}\bigvee_{l=1}^{n_i} (p^i_{j,l}(x,t) \bowtie^i_{j,l}\, 0) \Big) \\
        \equiv & \exists t \in \mathbb{R}^m.\,\, \forall x \in \mathbb{R}^n.\,\, \bigwedge_{i=1}^K \bigwedge_{j=1}^{m_i}\Big(\bigvee_{l=1}^{n_i} (p^i_{j,l}(x,t) \bowtie^i_{j,l}\, 0) \Big) \\
        \equiv & \exists t \in \mathbb{R}^m.\,\, \forall x \in \mathbb{R}^n.\,\, \bigwedge_{i=1}^K \bigwedge_{j=1}^{m_i} \\
        &\quad \Big((-p^i_{j,1}(x,t) \bowtie^i_{j,l}\, 0) \land \dots \land (-p^i_{j,1}(x,t) \bowtie^i_{j,n_i-1}\, 0) \Longrightarrow (p^i_{j,1}(x,t) \bowtie^i_{j,n_i}\, 0) \Big).
    \end{split}
    \end{equation*}
    The last formula yields an equisatisfiable system of PQEs with each PQE in the canonical form. Note that the above also yields a procedure for translating individual PQEs into their canonical forms. This concludes the proof. \hfill\qed
    
\end{proof}

\section{Example}

\begin{example} \label{example:schema}
	To illustrate how the problem of solving a system of PQEs arises in template-based synthesis for program analysis, we consider an example of proving termination of programs by computing ranking functions. Consider the program in Fig. \ref{example:schema} (left) and termination as a specification. We describe the three steps of the classical template-based method for synthesizing linear ranking functions~\cite{Colon03}. We consider linear programs and ranking functions for the simplicity of the example, however this method was extended to the setting of polynomial programs and ranking functions in~\cite{Chatterjee16} and is supported in \polyqent:
	\begin{compactenum}
		\item To find a linear ranking function, we first fix a symbolic linear expression template for each cutpoint location in the program:
		\[
		T_l(x) = \begin{cases}
			t_1 x + t_2 & \textit{ if ~} l = l_1 \\
			t_3 x + t_4 & \textit{ if ~} l = l_t
		\end{cases}
		\]
		where $t_1,t_2,t_3$ and $t_4$ are the symbolic template variables.
		\item A system of PQEs in Fig.~\ref{example:schema} (right) encodes that $T$ is a ranking function.
		\item Hence, any valuation of template variables $t_1,t_2,t_3,t_4$ that makes all PQEs valid gives rise to a correct ranking function for the program in Fig.~\ref{example:schema} (left).
	\end{compactenum} 
\end{example}

\begin{figure}[t]
	\centering
	\begin{subfigure}[t]{0.40\textwidth}
		\begin{lstlisting}[frame=none,numbers=none,escapechar=@,mathescape=true]
Invariant: 
  -1024 $\leq$ x $\leq$ 1023
$l_1$:  while x $ \geq$ 1 do
      x := x-1
   done
$l_t$:
		\end{lstlisting}
	\end{subfigure}
	\hfill
	\begin{subfigure}[t]{0.53\textwidth}
		\[
		\hspace{-4cm}
		\begin{split}
			&~~~\exists \blueit{t_1,t_2,t_3,t_4} \in \mathbb{R}\\
			&\begin{cases}
				\forall x \in \R, \Big(-1024 \leq x \leq 1023 \Rightarrow \blueit{t_1} x + \blueit{t_2} \geq 0 \Big) \\
				\forall x \in \R, \Big(-1024 \leq x \leq 1023 \wedge x \geq 1 \\ 
				~~~~~\Rightarrow \blueit{t_1} (x-1) + \blueit{t_2} \geq 0 \wedge \blueit{t_1} (x-1) + \blueit{t_2} \leq \blueit{t_1} x + \blueit{t_2} - 1 \Big) \\ 
				\forall x \in \R, \Big(-1024 \leq x \leq 1023 \wedge x < 1 \\~~~~~ \Rightarrow \blueit{t_3} x + \blueit{t_4} \geq 0 \wedge \blueit{t_3} x + \blueit{t_4} \leq \blueit{t_1} x + \blueit{t_2} -1 \Big)
			\end{cases}
		\end{split}
		\]
	\end{subfigure}
	\caption{A simple program (left) and the corresponding system of polynomial Horn clauses for computing ranking function that proves termination (right).}
	\label{fig:example}
\end{figure}

\section{Config File} \label{sec:config}

\begin{figure}
    \centering
    \includegraphics[width=\linewidth]{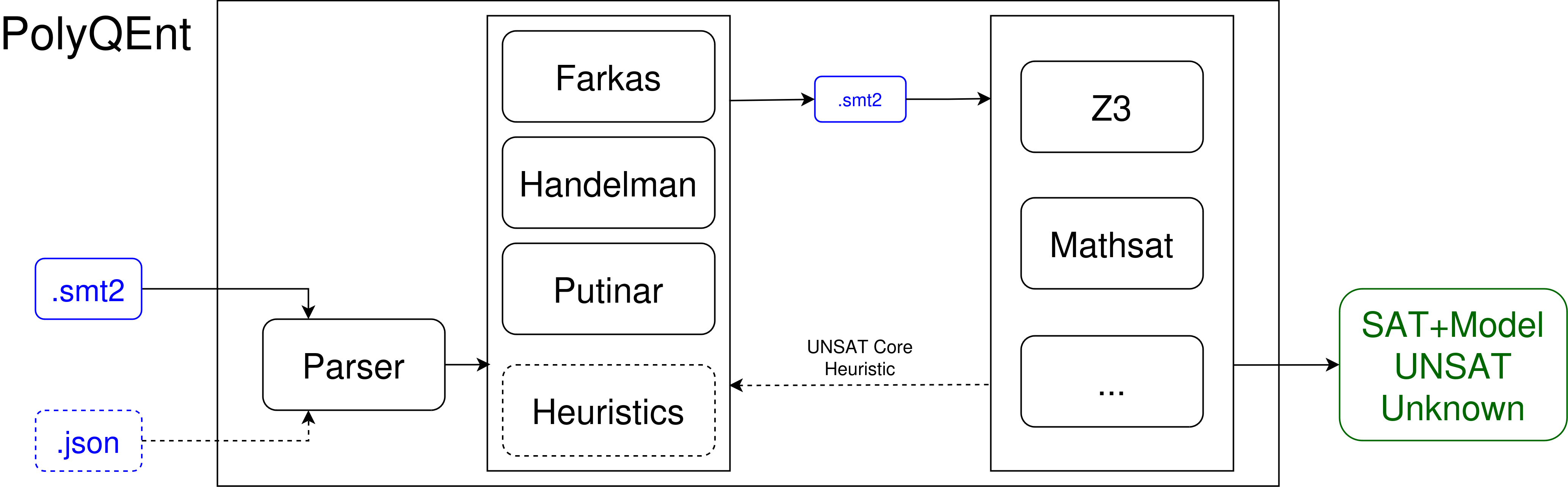}
    \caption{Architecture of \polyqent}
    \label{fig:enter-label}
\end{figure}

 \polyqent\ has a default (and recommended) configuration, which does not require the user to provide the config file. However, we also allow the user to change some of the parameter values used by \polyqent\ and the set of heuristics used by providing a \texttt{.json} config file:

\begin{enumerate}
	\item Positivity theorem to be used (\texttt{farkas}, \texttt{handelman} or \texttt{putinar}). The default configuration of \polyqent\ automatically chooses the most efficient positivity theorem to be applied while preserving soundness and relative completeness guarantees (see Section~\ref{sec:backend} for details). However, this optional parameter allows the user to opt for a different positivity theorem whose application is sound but not relatively complete, but may sometime lead to a more efficient constraint solving.
	\item Parameters of the positivity theorem to be used. See Appendix~\ref{subsec:geometry} for details. The default parameter values are also specified in Appendix~\ref{subsec:geometry}.
	\item The set of heuristics (if any) to be applied. The default configuration applies the Assume-SAT heuristic (see Section~\ref{sec:backend} for details).
	\item An SMT-solver to be used to solve the fully existentially quantified system of polynomial constraints resulting from applying the positivity theorem. The default configuration uses  \texttt{z3} \cite{de2008z3}, however \polyqent\ also supports \texttt{mathsat} \cite{cimatti2013mathsat5}.
	\item Background theory to be considered. The default is real arithmetic, however the config file allows the user to choose unbounded integer arithmetic.
\end{enumerate}

\section{Positivity Theorems}\label{subsec:geometry}

We now present the formal details behind the positivity theorems implemented in \polyqent. Provided are theorem statements, their use in \polyqent\ and tool parameters, the default configuration parameter values and how to change parameter values. All theorem claims are adopted from~\cite{AsadiC0GM21}, and we refer the reader to~\cite{AsadiC0GM21} for more details. Reading this section is optional for the users interested solely in running \polyqent\ with the default configuration.
 

\paragraph{Farkas' lemma.} We start by presenting Farkas' lemma, which is utilized by \polyqent\ when both the left-hand-side (LHS) and the right-hand-side (RHS) of PQEs in canonical form are specified in terms of linear inequalities.

\begin{theorem}[Farkas' lemma~\cite{farkas02}]
\label{thm:Farkas'}
    Consider a set $V = \{x_1, \ldots, x_n\}$ of real-valued variables and the following system of m linear inequalities over $V$
        \begin{equation*}
            \Phi := 
            \begin{cases}
                a_{1,0} + a_{1,1} \cdot x_1 + \ldots + a_{1,n} \cdot x_n \bowtie_1 0\\
                \hfil \vdots\\
                a_{m,0} + a_{m,1} \cdot x_1 + \ldots + a_{m,n} \cdot x_n \bowtie_m 0\\
            \end{cases}
        \end{equation*}
        where each $\bowtie_i\, \in \{>, \ge\}$. Exactly one of the following is true
        \begin{itemize}
            \item[F1)] $\Phi$ is satisfiable. Then, $\Phi$ entails the linear inequality $$\psi := c_0 + c_1 \cdot x_1 + \dots + c_n \cdot x_n \ge 0$$
        if and only if $\psi$ can be written as non-negative linear combination of the inequalities in $\Phi$ and $1 \ge 0,$ i.e.~if and only if there exist non-negative real-valued coefficients $y_0,\dots,y_m$ such that:
        $$
        \begin{matrix}
        c_0 = y_0 + \sum_{i=1}^m y_i \cdot a_{i, 0}\\
        c_1 = \sum_{i=1}^m y_i \cdot a_{i, 1}\\
        \vdots\\
        c_n = \sum_{i=1}^m y_i \cdot a_{i, n}.
        \end{matrix}
        $$
        \item[F2)] $\Phi$ is unsatisfiable and $-1 \ge 0$ can be derived as above.
        \item[F3)] $\Phi$ is unsatisfiable and $0 > 0$ can be derived as above, with at least one of the strict inequalities having strictly positive coefficient $y_i > 0$.
        \end{itemize}
        
        If the inequality $\psi := c_0 + c_1 \cdot x_1 + \dots + c_n \cdot x_n > 0$ is strict, then the case (F1) should be modified by requiring at least one coefficient $y_i$ with $\bowtie_i\, = \{>\}$ to be strictly positive, i.e.~$y_i > 0$.
\end{theorem}

Given a PQE in canonical form where both the LHS and the RHS are specified in terms of linear inequalities,  Farkas' lemma specifies the necessary and sufficient conditions for the LHS of the PQE to be unsatisfiable, i.e.~{\em (F2)} and {\em (F3)}, as well as for the LHS to be satisfiable and to entail the RHS, i.e.~{\em (F1)}. Hence, \polyqent\ translates each such PQE into a system of constraints defined by the Farkas' lemma. The application of Farkas' lemma in \polyqent\ introduces no new parameter values.

\paragraph{Handelman's theorem.} Next, we present Handelman's theorem. \polyqent\ utilizes this theorem when all inequalities on the LHS of a PQE in canonical form are linear, but the inequality on the RHS contains a polynomial of degree at least $2$. In order to formally present Handelman's theorem, we first need to define the notion of a monoid of a system of linear inequalities. Consider a set $V = \{x_1, \dots, x_n \}$ of real-valued variables and the following system of linear inequalities over $V$
     \begin{equation*}
        \Phi := \{
            f_1 \bowtie_1  0,
            \hfil \dots,
            f_m \bowtie_m 0
            \},
    \end{equation*}
where each $f_i = a_{i,0} + a_{i,1} \cdot x_1 + \ldots + a_{i,n} \cdot x_i$ and each $\bowtie_i\, \in \{>, \ge\}$. The {\em monoid} of~$\Phi$ of degree $d$, denoted $\text{\monoid}(\Phi, d)$, is the set of all polynomials over $V$ of degree at most $d$ that can be expressed as a product of linear expressions in $\Phi$, i.e.
\begin{equation*}
    \text{\monoid}(\Phi, d) := \Big\{\prod_{i=1}^m f_i^{k_i} | \forall i: k_i \in \mathbb{N}_0 \land \sum_{i=1}^m k_i \leq d\Big\}. 
\end{equation*}
We are now ready to formally state Handelman's theorem.

\begin{theorem}[Handelman's theorem~\cite{handelman88}]
\label{thm:handelman}
Consider a set $V = \{x_1, \dots, x_n \}$ of real-valued variables and the following system of $m$ linear inequalities over $V$
     \begin{equation*}
        \Phi := 
        \begin{cases}
            a_{1,0} + a_{1,1} \cdot x_1 + \ldots + a_{1,n} \cdot x_n \bowtie_1\,  0\\
            \hfil \vdots\\
            a_{m,0} + a_{m,1} \cdot x_1 + \ldots + a_{m,n} \cdot x_n \bowtie_m\, 0\\
        \end{cases}.
    \end{equation*}
where each $\bowtie_i\, \in \{>, \ge\}$. Suppose that $\Phi$ is satisfiable. Then, $\Phi$ entails polynomial inequality $g(x_1,\dots,x_n) \bowtie\, 0$ if
there exist a natural number $\bar{d}$, non-negative real numbers $y_0, y_1,\dots y_u$ 
and $h_1,\dots,h_u \in \monoid(\Phi, \bar{d})$ such that
$$
g = y_0 + \sum_{i=1}^u y_i \cdot h_i.
$$
Moreover, if the satisfiability set of $\Phi$ is bounded in $\mathbb{R}^n$, the linear inequalities $\bowtie_1,\dots,\bowtie_m$ are all non-strict and the polynomial inequality $\bowtie$ is strict, then $\Phi$ entails polynomial inequality $g(x_1,\dots,x_n) > 0$ {\em if and only if}
there exist a natural number $\bar{d}$, non-negative real numbers $y_0, y_1,\dots y_u$ 
and $h_1,\dots,h_u \in \monoid(\Phi, \bar{d})$ such that 
$g = y_0 + \sum_{i=1}^u y_i \cdot h_i$.
    
\end{theorem}

Given a PQE in canonical form where inequalities on the LHS are linear but the inequality on the RHS is at least a degree $2$ polynomial, Handelman's theorem specifies the sufficient conditions for the LHS to be satisfiable and to entail the RHS. 
On the other hand, the conditions for the LHS to be unsatisfiable are as in the Farkas' Lemma, i.e.~conditions {\em (F2)} and {\em (F3)} in Theorem~\ref{thm:Farkas'}. Hence, \polyqent\ translates each such PQE into a system of constraints defined by Handelman's theorem as well as {\em (F2)} and {\em (F3)} in Farkas' lemma. 

The translation is sound. The translation is also relatively complete, subject to the assumptions stated in Theorem~\ref{thm:handelman}, i.e.~that the satisfiability set of the inequalities on the LHS is bounded, the inequalities on the LHS are non-strict and the inequality on the RHS is strict.


\paragraph{\polyqent\ parameters for Handelman's theorem.} \polyqent\ defines a parameter for the maximal polynomial degree $\bar{d}$ of the monoid to be used in the translation. The value used by the \polyqent\ default configuration is the maximal polynomial degree appearing in the system of PQEs. In the config
file, this value can be changed by setting a new value of the \texttt{\em degree\_of\_sat} parameter.

\paragraph{Putinar's theorem.} Finally, we present Putinar's theorem, and its extension for the case when the LHS of the entailment is unsatisfiable which was proved in~\cite{AsadiC0GM21}. \polyqent\ utilizes these theorems when both the LHS and the RHS of a PQE in canonical form contain at least one polynomial of degree at least $2$. In what follows, we say that a polynomial $h$ is {\em sum-of-squares}, if it can be written as a finite sum $h = \sum h_j^2$ for squares of polynomials $h_j$.

\begin{theorem}[Putinar's theorem~\cite{putinar93}]
\label{thm:putinar}
    Consider a set $V = \{x_1, \dots, x_n \}$ of real-valued variables and the following system of m polynomial inequalities over~$V$
     \begin{equation*}
        \Phi := \{
            f_1(x_1, \dots, x_n) \bowtie_1  0,
            \hfil \dots,
            f_m(x_1, \dots, x_n) \bowtie_m 0
            \}.
    \end{equation*}
where each $\bowtie_i\, \in \{>, \ge\}$. Suppose that $\Phi$ is satisfiable. Then, $\Phi$ entails a polynomial inequality $g(x_1,\dots,x_n) \bowtie\, 0$ if
there exist positive real number $y_0$ and sum-of-squares polynomials $h_0,\dots,h_m$ such that
$$
g = y_0 + h_0 + \sum_{i=1}^m h_i \cdot f_i.
$$
Moreover, if the satisfiability set of at least one $f_i \geq 0$ is topologically compact (i.e.~closed and bounded), the linear inequalities $\bowtie_1,\dots,\bowtie_m$ are all non-strict and the polynomial inequality $\bowtie$ is strict, then $\Phi$ entails polynomial inequality $g(x_1,\dots,x_n) > 0$ {\em if and only if} there exist positive real number $y_0$ and sum-of-squares polynomials $h_0,\dots,h_m$ such that $g = y_0 + h_0 + \sum_{i=1}^m h_i \cdot f_i$.

\end{theorem}

\begin{theorem}[Unsatisfiability theorem~\cite{AsadiC0GM21}]
\label{thm:unsat}
    Consider a set $V = \{x_1, \dots, x_n \}$ of real-valued variables and the following system of $m$ polynomial inequalities over $V$
     \begin{equation*}
        \Phi := \{
            f_1(x_1, \dots, x_n) \bowtie_1  0,
            \hfil \dots,
            f_m(x_1, \dots, x_n) \bowtie_m 0
            \}.
    \end{equation*}
    where each $\bowtie_i \in \{>, \ge\}$. Then $\Phi$ is unsatisfiable if and only if at least one of the following two conditions holds:
    \begin{itemize}
        \item[U1)] There exist a positive real number $y_0$ and sum-of-square polynomials $h_0, \dots, h_m$ such that
        $$-1 = y_0 + h_0 + \sum_{i=1}^{m} h_i \cdot f_i$$
        \item[U2)] There exist $d \in \mathbb{N}_0$ and polynomials $h_1', \dots, h_m'$ over $V \cup \{w_1, \dots, w_m\}$, such that for some $j$ in which $\bowtie_j \in \{>\}$, $w_j^{2\cdot d} = \sum_{i=1}^{m} h_i' \cdot (f_i - w_i^2)$.
    \end{itemize}
\end{theorem}

Given a PQE in canonical form where both the LHS and the RHS contain a polynomial of degree at least $2$, Putinar's theorem specifies the sufficient conditions for the LHS to be satisfiable and to entail the RHS. 
On the other hand, the conditions for the LHS to be unsatisfiable are given by conditions {\em (U1)} and {\em (U2)} in Theorem~\ref{thm:unsat}. Hence, \polyqent\ translates each such PQE into a system of constraints defined by Putinar's theorem and Theorem~\ref{thm:unsat}. 

The translation is sound. The translation is also relatively complete, subject to the assumptions stated in Theorem~\ref{thm:putinar}, i.e.~that the satisfiability set of at least one inequality on the LHS is topologically compact, that the inequalities on the LHS are non-strict and the inequality on the RHS is strict.


\paragraph{\polyqent\ parameters for Putinar's theorem.} \polyqent\ defines the following parameters when invoking Putinar's theorem. In the default configuration, the values of all these parameters are set to the maximal polynomial degree appearing in the system of PQEs, and the values of these parameters can be changed by modifying the config file:
\begin{itemize}
    \item Maximum degree of $h_i$'s in Theorem~\ref{thm:putinar}. In the config file, \texttt{\em degree\_of\_sat} represents this parameter.
    \item Maximum degree of $h_i$'s in Theorem~\ref{thm:unsat} (condition {\em (U1)}). In the config file, \texttt{\em degree\_of\_nonstrict\_unsat} represents this parameter.
    \item Maximum degree of $h_i$'s in Theorem~\ref{thm:unsat} (condition {\em (U2)}). In the config
    file, \texttt{\em degree\_of\_strict\_unsat} represents 
    this parameter.
    \item Value of $d$ in Theorem~\ref{thm:unsat} (condition {\em (U2)}). In the config file, \texttt{\em max\_d\_of\_strict} represents this parameter.
\end{itemize}


\begin{remark}[Assume-SAT]
\label{remark:assumesat}
    Heuristic {\em Assume-SAT} removes the conditions that consider the case when the LHS of a PQE is unsatisfiable in all the previous theorems.
    If Farkas' lemma or Handelman's theorem is applied, \polyqent\ omits the conditions  {\em (F2)} and {\em (F3)} in Theorem~\ref{thm:Farkas'}. If Putinar's theorem is applied, \polyqent\ omits the conditions {\em (U1)} and {\em (U2)} in Theorem~\ref{thm:unsat}.
\end{remark}
    